\documentclass[times]{article}
\usepackage{mhchem}
\usepackage{times,mathptm}
\usepackage{subfigure}
\usepackage{dblfloatfix}
\usepackage{xcolor}
\usepackage{graphicx}
\usepackage{amssymb}
\usepackage{setspace}
\usepackage[hidelinks]{hyperref}
\usepackage{siunitx}
\usepackage{amsmath}
\usepackage{amsthm}
\newtheorem{theorem}{Theorem}

\DeclareMathOperator\erf{erf}
\usepackage{lineno,hyperref}

\usepackage{siunitx}
\begin{document}



\huge{\noindent Thermal conductance of interleaving fins}


\vspace{0.6cm}
\normalsize
\noindent Michiel A. J. van Limbeek$^{a,b}$ \footnote{Corresponding author: michiel.vanlimbeek@ds.mpg.de} and Srinivas Vanapalli$^b$\\
$^a$\textit{Max Planck Institute for Dynamics and Self-Organization, 37077 G\"ottingen, Germany}\\
$^b$\textit{University of Twente, Postbus 217, 7500 AE Enschede, The Netherlands}\\

\section*{Abstract}
Interleaving fins can significantly increase the heat transfer by increasing the effective area per unit base area. The fins are separated uniformly by a gap, which is filled with a flow medium to control the heat flux. The heat flux between the plates depends strongly on the thermal conductivity of the fin material and the medium between them as well as the dimensions. In earlier studies empirical a fitting method is used to determine the total effectiveness of the fins. However, it required complete characterization of the fins for each new set of operating conditions. In this paper, a simplified analytical model, but still preserving the main physical traits of the problem is developed. This model reveals the dimensionless parameter group containing both material properties and the fin geometry that govern the heat transfer. Rigorously testing of the model using a numerical finite element model shows an accuracy within 2 \% over a large parameter space, varying both dimensions and material properties. Lastly, this model is put to test with previously measured experimental data and a good agreement is obtained.



\section{Introduction}

Temperature control is essential in many situations: from regulating the body temperature of biological species to keeping your drinks cold in your fridge and preventing your computer from overheating. Whereas rapid heating can be easily achieved by resistive electrical heating, rapid cooling is more complicated as heat needs to be transported towards a heat sink. Many cases use conduction for transport for which the thermal conductance is constant. For some cases however, this is not desirable and one wants to regulate the conductance ie change it in time.
  
As an example in hyperpolarized MRI molecular imaging systems, the sample sleeves that accommodate variable temperature inserts are integrated in the cryostat structure, requiring the entire system to be warmed up for servicing. In a recent development a heat switch was used to thermally disconnect the sample sleeve from the cold plate for quick servicing \cite{Stautner_2019}. Another common example is a thermal battery, where some (phase change) material \cite{medrano2009experimental,castell2011maximisation,alawadhi2000performance,farid2004review} is used to maintain  a desired operating temperature window for the instrument. The battery should be thermally coupled to a cold sink to charge the system \cite{BONFAIT2009326}. In situ cryofixation of cells is of growing interest in phase‐contrast microscopy for studying dynamic cellular processes associated with physiological and pathophysiological conditions \cite{schneider1995cryo,schneider1998cryo}. Moreover, many space instruments require that the cryocooler system possess a very high level of reliability. This need for high reliability requires that some form of redundancy be incorporated. One common implementation is standby redundancy. Until the primary unit breaks down the thermal connection to the redundant cooler should be low \cite{ross2002cryocooler}. However, when the standby unit is needed in case of the break down of the primary unit, the thermal conductance should be high. A similar thermal behavior is essential for the connection between equipment and space radiators \cite{bugby2010two,glaister1996application,gross1970thermal}. 

Unifying these examples is the use of a thermal switch, whose conductance can be changed rapidly. In the off-state, the heat transfer is minimized, whereas in the on-state it is maximized in order to exchange heat rapidly with the heat sink/source.  Such a switch can be constructed by having two plates with a small gap between them \cite{catarino2008neon}. The conductivity across the gap increases dramatically when the content is changed from a vacuum by a gas \cite{hamburgen1998interleaved} or by replacing it with a liquid.  The gas pressure can be controlled actively or passively by an evaporator \cite{romera2010control} or sorption cells \cite{catarino2008neon, burger2007long}.

To increase the heat flux density per base area, interleaving fins can be used. Rows of fins are attached on the two aforementioned plates in a staggered configuration. Previous research on interleaved fins explored the effective thermal resistance across the base plates as a function of the fin dimensions assuming uniform temperature along the length of the fin. This assumption require high thermal conductivity of the fin compared to the gas, and a further requirement is to have a low fin aspect ratio. Studies were done at both cryogenic \cite{shirron2005passive, dipirro2014heat} and room temperature \cite{vanapalli2015passive,krielaart2015compact}.

Although fins are an efficient way to increase heat exchange to the surroundings, they can loose its efficiency when the fin cools down in the longitudinal direction. This effect can be characterized by the Biot number of the fin: $Bi=h L /k$ where $h$ is the heat transfer coefficient, $L$ the length of the fin and $k$ the thermal conductivity. For $Bi\ll 1$ the conduction in the fin can easily supply the heat towards the cooling edges and no gradients occur, whereas for $Bi\gg 1$ the opposite holds and the fin has a non-uniform temperature profile. Although here the heat transfer coefficient is associated with convective transport towards the far-field temperature $T_\infty$, we expect that in the present case of two opposing staggered fin arrays a similar behaviour can be expected. Indeed, when replacing $h$ by purely conduction over a gas filled gap we obtain $Bi^*=k_\mathrm{g}L /(k_\mathrm{s} D)$, where $D$ denotes the (uniform) gap thickness and the subscripts g and s refer to the gas and solid domains. What differs however from the classical fin is that here $T_\infty$ is now the \textit{local} temperature of the opposing fin, in this way coupling the two sides of the thermal link.

In this study we develop an analytical one-dimensional  model, which is validated by numerical solutions of the problem. We identify a more accurate Biot number and find that indeed the cooling starts to be relevant when this number is of the order unity. Several ways are then explored to test the applicability of the model and for which parameters it suffers accuracy. We also provide an expression for the heat flux across the thermal link, which is more relevant for the design.  Finally we test the gap experimentally in a range where the fins suffer from cooling, finding good agreement with our analytical prediction.



\section{Problem description and Approach}
The thermal link is made up by two opposing plates with staggered fins. In between them is a gap of width $\hat D$, which  filled with a gas of conductivity $\hat{ k}_g$, where the hat is used to indicate dimensional quantities. The hot plate and cold plate of the stack is separated by a distance $\hat{ \mathcal{L}}$ and the problem is quasi-2D. To increase the effective area, fins of width $\hat W$ are attached on both plates in the space between them. Although in general the spacing and length are free to choose, we here restrict ourselves to the case of a constant separation $\hat D$. The fins start from a position $\hat \delta$ from the top or bottom of the stack.  The region where heat is exchanged is then $\hat L=\hat{\mathcal{L}}-2\hat \delta$. Figure \ref{fig:domain} shows this geometry. Many fins are positioned next to each other in the y-dimension  with spacing $2\hat W$. The hot plate is at $\hat T=\hat{T}_h$ and the cold plate at $\hat T=\hat{T}_c$. We are interested on the heat transfer between the plate, which have a conductivity $\hat{k}_s$. We identified thus four geometrical parameters, namely $\hat{\mathcal{L}}$, $\hat \delta$, $\hat W$, $\hat D$ and three parameters related to the heat transfer: the conductivities of the gas $\hat{k}_{g}$ and solid $\hat{k}_{s}$ and the temperature difference between the hot and cold plate. 
\begin{figure}
\includegraphics[width=\columnwidth]{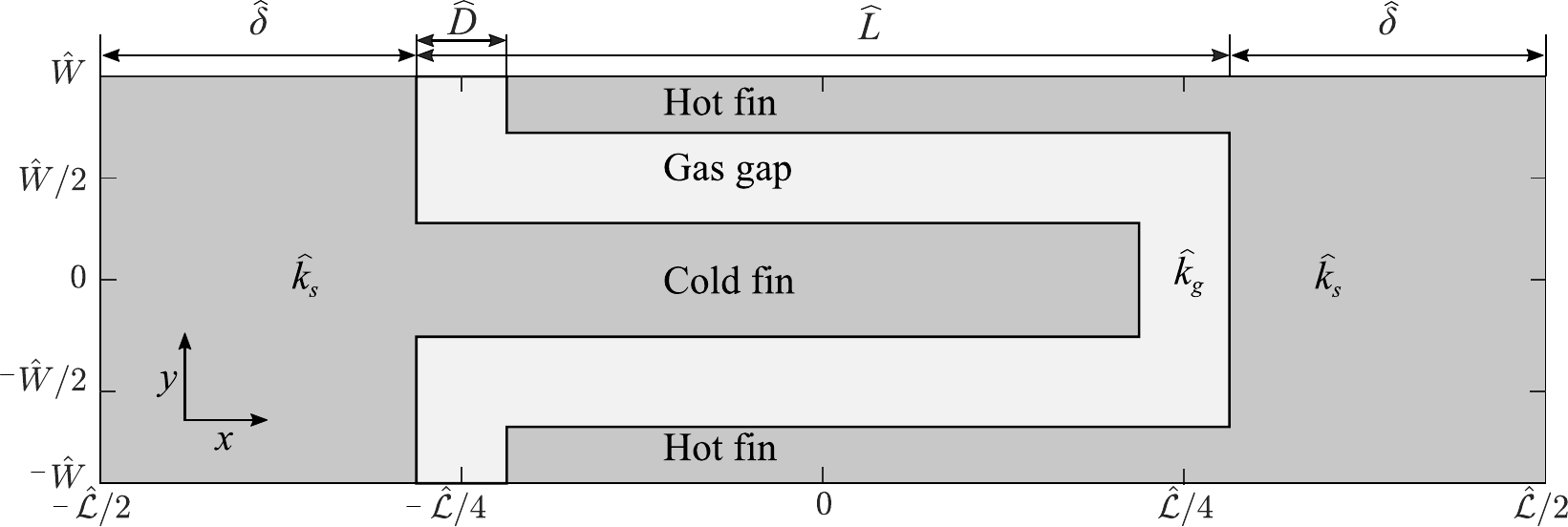}
\caption{Sketch of the geometry of a single pair of fins}
\label{fig:domain}
\end{figure}
Let us now outline the procedure to describe the heat transfer analytically, without the need of a full quasi- two dimensional numerical simulation.
\paragraph{Isothermal fins} Similar to the classical problem of a single fin in an infinite fluid (see \cite{Bejan} 2.7), we expect in the limit of $\hat{k}_\mathrm{s} \rightarrow \infty$ no gradients to emerge into the fin. The heat exchanged $\hat Q$ is then simply $-\hat{k}_\mathrm{g} \hat A \Delta \hat T /\hat D $, with $\hat A$ the fin area. Consider a array of interleaving fins on a plate of size $\hat Z$ times $\hat Y$. One finds for the total area $\hat A=\frac{\hat Y}{2\hat W}2(\hat W+\hat{\mathcal{L}}-2\hat{\delta} -\hat{D})\hat Z$ and obtains for the heat flux $\hat q''=\hat Q/(\hat Z \hat Y)$:
\begin{equation}
\hat q''= -(\hat T_h-\hat T_c)\hat k_g \frac{\hat W+\hat{\mathcal{L}}-2\hat \delta-\hat D}{\hat D\hat W}.
\label{eq:isothermalflux}
\end{equation} 
This solution will be used as a reference to compare the case in which the fin is subject to non-isothermal effects, which will be discussed now.
\paragraph{Non isothermal fins}
Let us now sketch the solution for non-isothermal fins. The slenderness of the fins prompts one to use a one-dimensional model, similar to the classical fin problem. We study a single pair of fins, for which we find a profile $\hat T(x)$ for the temperature in the fin at the cold plate. Similarly for the fin at the hot plate we have the profile $\hat{\mathcal{T}}(\hat x)$. We will then use the temperature difference $\hat{\Theta} (\hat x)$ between them, as the heat exchanged at any position in the fin is proportional to $\hat \Theta$. The problem is described by two coupled heat equations, which require four boundary conditions to solve.  The problem will be solved in the following section.
\section{Modelling cooling in the fins }
\label{sec:math}
Depending on the length scales and material properties, we explore the most simple solution to the problem, without introducing too many errors compared the full 2-D (numerical) calculation. 
Let us first investigate the conditions where a one dimensional approach is applicable. Since we here study a steady-state situation without any heat sources, the heat equation in the domain reads 
\begin{equation}
\nabla \cdot \vec{j}=0,
\label{heateq}
\end{equation}
where $\vec{j}=\hat{k}\nabla \hat{T}$; the heat flux which is modelled using Fourier's law. In the current study we assume $\hat{k}$ to be independent of the temperature. Considering a single fin now, we can rescale the $\hat{x}$ and $\hat{y}$ coordinate with the length $\hat{\mathcal{L}}$ and width $\hat{W}$ of the fin as $\hat{x}=\hat{\mathcal{L}}\bar{x}$ and $\hat{y}=\hat{W}\bar{y}$ ,= obtaining:
\begin{equation}
\hat{ k}_s \left(\partial_{\hat{x}\hat{x}}\hat{T}+\partial_{\hat{y}\hat{y}}\hat{T}\right)=\hat{ k}_s\left(\partial_{\bar{x}\bar{x}}\hat{T}+\frac{\hat{\mathcal{L}}^2}{(\hat{W}-\hat{D})^2}\partial_{\bar{y}\bar{y}}\hat{T}\right)=0,
\end{equation} 
which acts on the domain $x,y=[0,1]$. Inspection shows that for large $\frac{\hat{\mathcal{L}}^2}{(\hat{W}-\hat{D})^2}$ the gradients in the $\bar{y}$ direction become small compared those in the $\bar{x}$ direction. This prompt us to assume $\hat{T}(\hat{x},\hat{y})\approx \hat{T}(\hat{x})$, which we will validate a posteriori. Equation \ref{heateq} is now integrated over a small slice of the fin. The slice extends over the full width of the fin $\hat W-\hat D$ and has a thickness $\Delta \hat{x}$ so we obtain using Gauss-theorem in 2 dimensions:
\begin{align}
\int_v \nabla \cdot \vec{j} \; \mathrm{d}V=\oint \vec{j}\cdot \vec{n}  \; \mathrm{d}\ell&=\; 0& \nonumber\\
\hat{ k}_s(\hat{W}-\hat{D})\left(\frac{\mathrm{d}\hat{T}(\hat{x}+\Delta \hat{x})}{\mathrm{d}{\hat{x}}}- \frac{\mathrm{d}\hat{T}(\hat{x})}{\mathrm{d}{\hat{x}}} \right)+2\Delta\hat{x} \hat{q}_\mathrm{gap}&=\;0.&
\end{align}
 Evaluating the limit of $\Delta \hat{x}\rightarrow 0$ yields,
\begin{equation}
\hat{ k}_s (\hat{W}-\hat{D})\frac{\mathrm{d}^2\hat{T}(\hat{x})}{\mathrm{d}\hat{x}^2 }+ 2\hat{q}_\mathrm{gap}=0,
\end{equation}
where a Taylor expansion was used. Here $\hat{q}_\mathrm{gap}$ is the heat exchanged with the opposing fin, which we model as $-\frac{\hat{k}_g}{\hat D} (\hat{T}(\hat{x})-\hat{\mathcal{T}}(\hat{x}))$. Using a similar procedure for the opposing fin gives us the following coupled equations for the temperature profiles $T(\hat{x})$ and $\mathcal{T}(\hat{x})$:
\begin{align}
\hat{ k}_s(\hat{W}-\hat{D})\frac{\mathrm{d}^2\hat{T}}{\mathrm{d}\hat{x}^2 }-\frac{2\hat{k}_g}{ \hat D} (\hat{T}(\hat{x})-\hat{\mathcal{T}})&=0\quad \mathrm{and} \nonumber \\
\quad \hat{ k}_s(\hat{W}-\hat{D})\frac{\mathrm{d}^2\hat{\mathcal{T}}}{\mathrm{d}\hat{x}^2 }-\frac{2\hat{k}_g}{ \hat D} (\hat{\mathcal{T}}-\hat{{T}})&=0\quad .
\label{eq:doubleT}
\end{align}
 The coupled system of equations is subject to the following boundary conditions: The temperature profiles should satisfy the Dirichlet boundary condition $\hat{T}(-\mathcal{L}/2)=\hat{T}_c$ or $\hat{\mathcal{T}}(\mathcal{L}/2)=\hat{T}_h$ respectively. The problem is closed by noting that heat is conserved in the system, from which we find: $\left. k_s \frac{\mathrm{d}\hat{T}}{\mathrm{d}x}\right|_{-\mathcal{L}/2}=k_s\left. \frac{\mathrm{d}\hat{\mathcal{T}}}{\mathrm{d}x}\right|_{\mathcal{L}/2}$.

It is interesting to investigate similarities with section 2.7.4 in reference Bejan\cite{Bejan} to validate the one dimensional model.  The present model is valid as long as the Biot number for the width of the fin    $Bi^{\dag}=h(\hat W-\hat D)/\hat{k}_s=\hat k_g(\hat W-\hat D)/(\hat{k}_s \hat D)\ll 1$. Here we used $h=\hat{k}_g/\hat{D}$ for the heat transfer coefficient over the conducting gas gap.
As long as $Bi^{\dag}=hW/k=k_gW/(kD)<1$, no gradients occur in the cross sectional plane. 
We now non-dimensionalize the spacial dimension and lengths by the thickness $\hat{\mathcal{L}}$ of the stack: $ \hat x =x \hat{\mathcal{L}}$ and the temperatures using the boundary temperatures (\textit{e.g.} $\hat T=(\hat T_h-\hat T_c) T +\hat T_c$). For thermal conductivities we rescale by the conductivity of the solid: $k=\hat{k}_g/\hat{k}_s$. Let $\Theta=T-\mathcal{T}$, and $C^2= \frac{4\hat{k}_g \hat{\mathcal{L}}^2}{\hat{k}_s \hat D( \hat W-\hat D)}$ we then obtain 
\begin{equation}
\label{eq:Theta}
\Theta''-C^2\Theta=0, 
\end{equation}
and 
\begin{equation}
\label{eq:T-Theta}
 T''-\frac{1}{2}C^2 \Theta=0, \quad \mathcal{T}''+\frac{1}{2}C^2 \Theta=0,
\end{equation}

where we used primes to represent the derivatives. Let the centre of the stack be at $\hat x=0$. We then find for the boundary conditions:
$T(-1/2)=0,$ $\mathcal{T}(1/2)=1$ and $T'(-1/2)=\mathcal{T}'(1/2)$.
The solution for $\Theta$ is then $\alpha \exp{({C}x)}+\beta \exp{(-{C}x)}$. 
\begin{theorem}
The symmetry of the geometry results in $\Theta$ to be an even function.
\end{theorem}
\begin{proof}
Let $\Theta$ and $\tilde{\Theta}$ be solutions to Eq. 3 and  $\tilde{\Theta}(x)=\Theta(-x)$. We then find $\Theta(-x)''-C^2\Theta(-x)=0\rightarrow \tilde{\Theta}(x)''-C^2\tilde{\Theta}(-x)=0$. At $x=0$ the two solutions coincide: $\Theta(0)=\tilde{\Theta}(0)$ and hence from uniqueness we deduce $\Theta(x)=\tilde{\Theta}(x)=\Theta(-x)$ to be an even function.
\end{proof}
In order for $\Theta$ to be even we find $\alpha=\beta\equiv 2 \Theta_0$, which yields:
\begin{equation}
\Theta=\Theta_0\cosh\left( {C}x\right),
\label{eq:Thetasolve}
\end{equation}
Substitution of \autoref{eq:Theta} into \autoref{eq:T-Theta} yields for $T$ after integrating twice:
\begin{equation}
T=\frac{1}{2}\Theta_0\cosh\left( {C}x\right)+a+b\;x,\label{eq:Tfin}
\end{equation}
where we find three integration constants: $a,b$ and $\Theta_0$. First, we use continuity of flux in a single fin to find per pair of opposing fins, \textit{i.e.} per $2W$:
\begin{equation}
-(W-D) T'|_{-\frac{1}{2}+\delta}=2\int_{-\frac{1}{2}+\delta}^{\frac{1}{2}-\delta} \frac{k}{D}\Theta\; \mathrm{d}x-(W-D)\; T'|_{\frac{1}{2}-\delta}.
\end{equation}
The left-hand side denotes the flux at the base of the fin, the integral the flux crossing the gap sideways and the last term the heat exchange at the tip of the fin. Using \autoref{eq:Tfin} one obtains
\begin{equation}b=\Theta_0\frac{{C}}{2} \sinh{(\mathcal{C})}-\frac{2}{W-D} \frac{k}{D}\Theta_0\frac{2}{{C}}\sinh{(\mathcal{C})}-\frac{k}{D}\Theta_0\cosh(\mathcal{C}), \nonumber
 \end{equation}
where $\mathcal{C}=C(\frac{1}{2} -\delta)$. Next we approximate the flux through the support of the fins as $ 2 W  \left(T(-\frac{1}{2})-T(-\frac{1}{2}+\delta)\right)/\delta$. This flux is equal to the sum of the heat entering the fin and directly flowing from the base into the tip of the opposing fin: $-(W-D) T'|_{-\frac{1}{2}+\delta}+(W-D)\frac{k}{D} \Theta(-\frac{1}{2}+\delta).$ This way we can express $a$ as:
\begin{align}
\label{eq:a}
a=&-\frac{1}{2}\Theta_0\cosh{(\mathcal{C})}-b\left(-\frac{1}{2}+\delta\right) \dots\nonumber\\
&-\delta\frac{W-D}{2W}\left(\Theta_0\frac{{C}}{2}\sinh(\mathcal{C})-b\right)+\delta\frac{W-D}{2W}\frac{k}{D}  \Theta_0 \cosh(\mathcal{C}).
\end{align}

Finally, we can close the problem by evaluating the solution at $x=-1/2$, where we find $T=0$, to obtain the value of $\Theta_0$. Note that $\Theta_0$ is the temperature difference between the fins at $x=0$. The symmetry of the two opposing fins allows us to write $\Theta_0=2(T(x=0)-1/2)$, where $1/2$ represents $(T_h+Tc)/2$. We then find $a=1/2$, from which one can obtain $\Theta_0$ more easily using \autoref{eq:a}. $\Theta_0$ is therefore a measure for the amount of cooling in the system: in the isothermal limit we find $\Theta_0=-1$.

We now arrive at the following equation for the temperature field:

\begin{align}
\label{eq:threepiecessolution}
T=
\begin{cases}
-(x+\frac{1}{2})\; \frac{W-D}{2W}\left(\Theta_0\frac{{C}}{2} \sinh{(\mathcal{C})}-b\right)&\text{\small : $ -\frac{1}{2}\leqslant x<-\frac{1}{2}+\delta$}\\
\text{\autoref{eq:Tfin}} &\text{ \small : $  -\frac{1}{2}+\delta \leqslant x< \frac{1}{2}-\delta$} \\
1-(x+\frac{1}{2})\; \frac{W-D}{2W}\left(\Theta_0\frac{{C}}{2} \sinh{(\mathcal{C})}-b\right)&\text{ \small  :  $-\frac{1}{2}-\delta\leqslant x<\frac{1}{2}$}
\end{cases}
\end{align}

whose dimensional form can be obtained easily. A similar procedure can be followed to obtain $\mathcal{T}$.  
\paragraph{Heat flux}
From an application perspective, the quantity of interest is the (dimensional) heat exchanged per frontal unit area. We find for the width of a single pair of fins:
\begin{equation}
\label{eq:dim"lessflux}
2W q''= \frac{k }{D}\Theta_0 \left( \frac{4}{{C}}\sinh{(\mathcal{C})}+2(W-  D)\cosh{(\mathcal{C})} \right),
\end{equation}

or in dimensional units
\begin{equation}
\label{eq:dimflux}
2\hat{W}\hat{q}''=\Delta \hat T \frac{\hat{k}_g \hat{\mathcal{L}}}{\hat{D}}\Theta_0 \left( \frac{4}{{C}}\sinh{(\mathcal{C})}+2\frac{\hat W- \hat D}{\hat L}\cosh{(\mathcal{C})} \right), 
\end{equation}
 from which $\hat{q}''$ can be obtained. It is interesting to expand \autoref{eq:dimflux} in the limit of $C \rightarrow 0$ as it reduces correctly to \autoref{eq:isothermalflux}: $\hat q''=\Delta \hat T \frac{\hat{k}_g \hat{\mathcal{ L}}}{2\hat W \hat D}\Theta_0 \left( 4 \cdot\left( \frac{1}{2}-\frac{\hat{\delta}}{\hat{\mathcal{L}}}\right)+2\frac{\hat W- \hat D}{\hat {\mathcal{ L}}} \right)+\mathcal{O}(C^2)+\dots$.
 %
\section{Results}
The analytical model is first validated using numerical solutions  of the problem using the commercial package Comsol \cite{Comsol}.  The steady state heat equation was solved by a finite element method on a triangular grid. We consider one half of a pair of opposing fins inside a large stack, avoiding any edge effects. The hot and cold faces have Dirichlet boundary conditions where the longitudinal boundaries satisfy the no-flux condition. The solid-gas boundaries satisfy continuation of heat flux and a no-jump condition in temperature.

\begin{figure}
\centering
\includegraphics[width=\columnwidth]{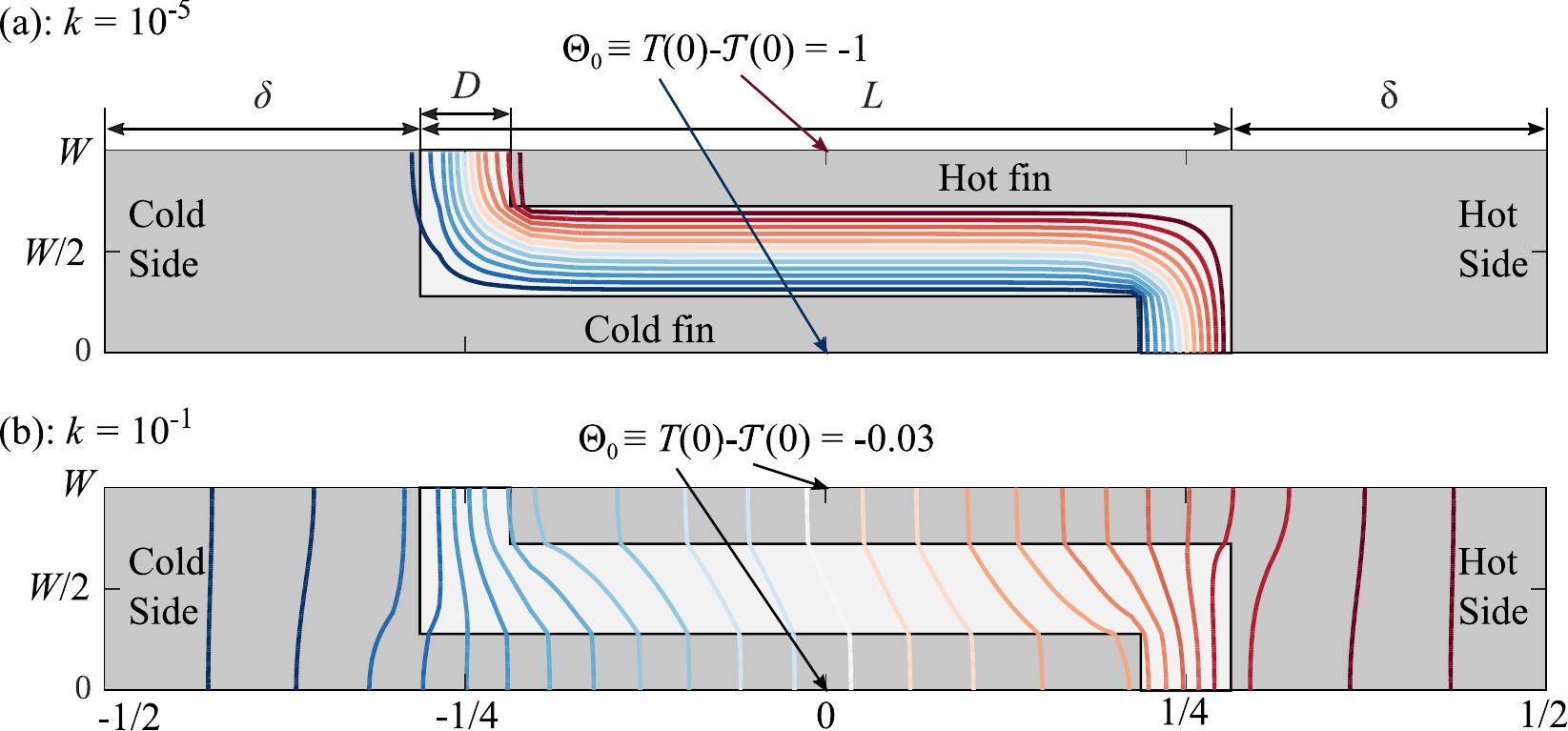}
\caption{Numerical solution for $D=0.0625$, $\delta=0.22$ and $W=0.14$. Results for a low conducting gas of $k=\hat{k}_g/\hat{k}_s=10^{-5}$ is displayed in panel (a), whereas panel (b) was solved with $k=10^{-1}$, resulting in $C=0.9 $ and $9.1$ resp. The coloured lines indicate isotherms between $T=0$ (blue) and 1 (red). }
\label{fig:plotexample}
\end{figure}

Two cases are presented in \autoref{fig:plotexample} for  $C=0.9 $ and $9.1$, where the gas conductivity is varied. It can clearly be seen that the isotherms in \autoref{fig:plotexample}(a), shown by the coloured lines, lay mostly within the gas gap. Strong gradients appear in the panel (b) however, as expected since $C>1$. In the centre, the temperature difference between the two profiles $\Theta_0$ is added. We find that for isothermal fins it is close to unity, while decreasing when significant gradients occur in the film since the conductivity of the gas increases. 

We test now how trustworthy $C$ is in predicting the occurrence of non-isothermal fins. We vary the parameters $W$, $\hat{\mathcal{L}}$, $D$, $\hat{k}_s$ and $\hat{k}_g$ independently and study the behaviour of $\Theta _0$, which is presented in \autoref{fig:ResultCTheta}. Despite through which parameter $C$ is changed, the behaviour of $\Theta_0$  is the same. We clearly find that for $C<0.1$ the fins are isothermal, whereas for $C=[0.1,1]$ it only decreases by a fifth at most.   Increasing $C$ further results quickly in a decrease in effectiveness where from $C\approx10$ no significant temperature difference between the fins is found. 
The behaviour of $\Theta_0$ was accurately fitted using a single fitting parameter $b$:
\begin{equation}
\Theta_0=\frac{1}{2}\left( \erf \left( \log_{10}\left( C^2 \right) +b \right) -1\right),
\label{eq:theta0fit}
\end{equation}
as presented in \autoref{fig:ResultCTheta}, where we found $b=-0.7$.
\begin{figure}
\centering
\includegraphics[width=0.9\columnwidth]{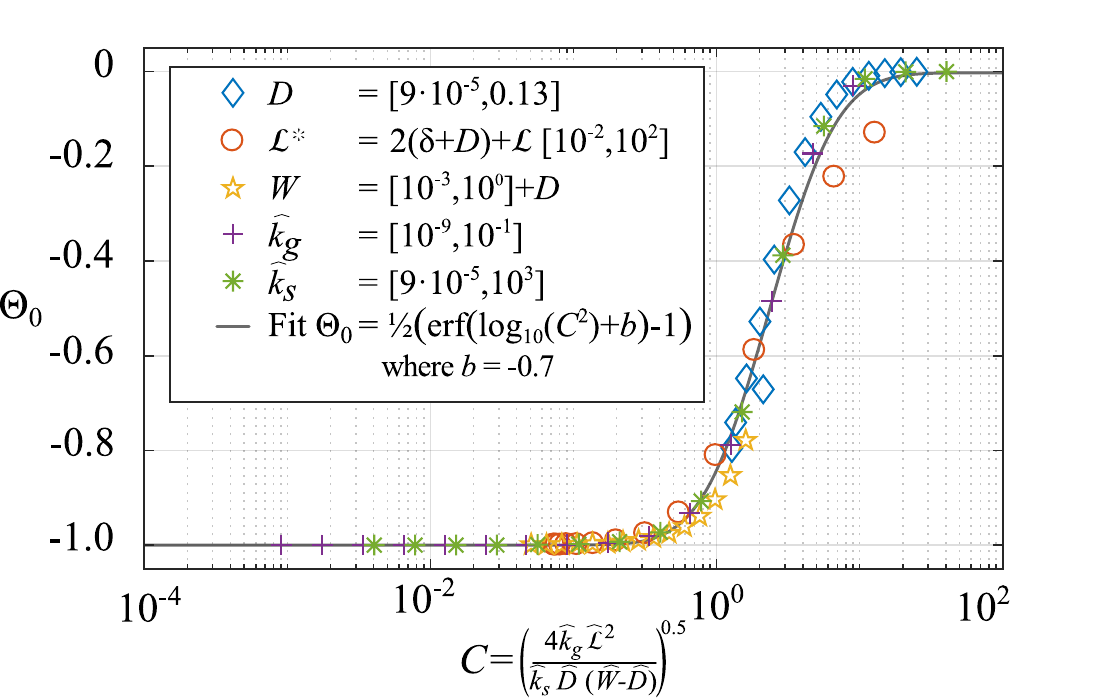}
\caption{Temperature difference $\Theta_0$ between the two fins. Rescaling the parameters using $C$ reveals a single curve for which \autoref{eq:theta0fit} was fitted (black line). }
\label{fig:ResultCTheta}
\end{figure}

\begin{figure}
\centering
\includegraphics[width=\columnwidth]{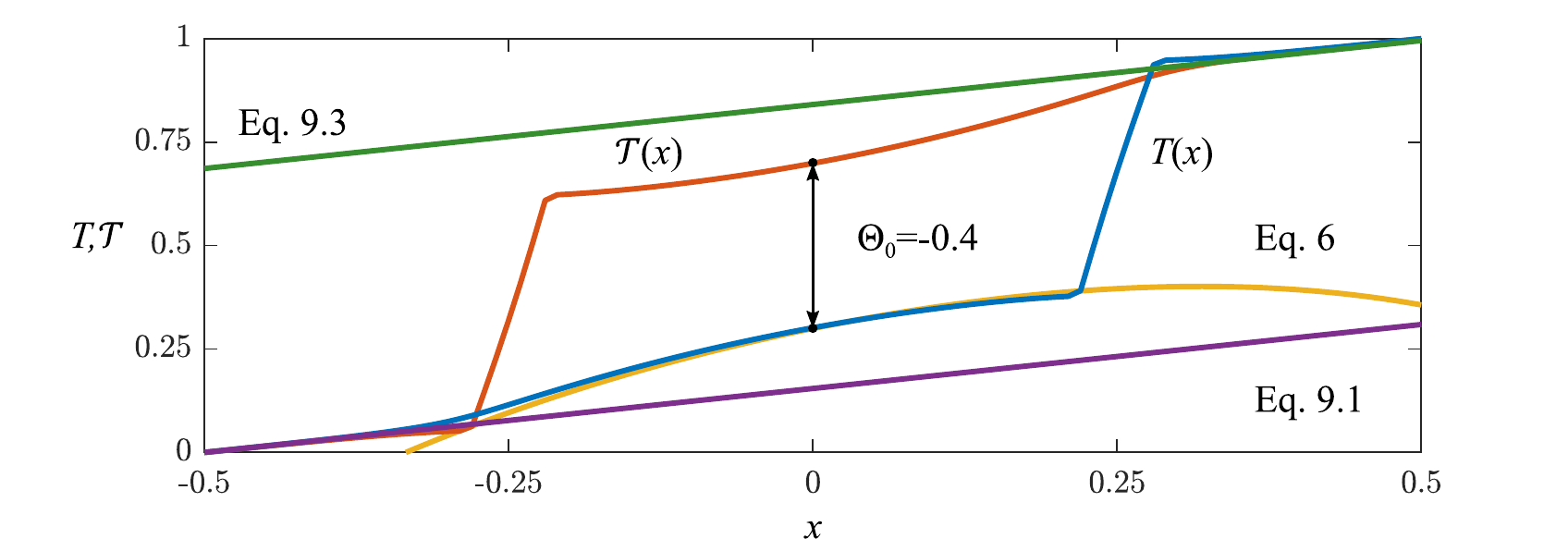}
\caption{Comparison between numerical solutions of $T$ (blue) and $\mathcal{T}$ (red) and the analytical model for $C=2.9$. $\Theta_0$ was evaluated to be $-0.4$. }
\label{fig:mathexample}
\end{figure}

So far we have focussed purely on numerical results, so let us now focus on comparing the solutions to the model developed earlier in  \autoref{sec:math}. From the numerical solution for $C=2.9$ and $\delta=0.22$ we obtain the profiles in the hot and cold fin, as shown in \autoref{fig:mathexample}. Focussing on the the blue temperature profile $T$ of the cold fin, we use \autoref{eq:threepiecessolution}, plotted in purple, yellow and green. Although the agreement looks good for all pieces, we quantify this by comparing the following properties: first we compare the difference in $C$ and $\Theta_0$ by fitting the temperature difference of the numerical solution using \autoref{eq:Thetasolve}. We compare the fit then with the expression for $C$ and $\Theta_0$. Secondly we compare the difference of the fin profiles between the numeric- and analytical model using the relative error
\begin{equation}
e_{fin}=\frac{1}{\Delta x}\frac{2}{\bar{T}_\mathrm{num}+\bar{T}_\mathrm{Eq. 6}} \sum_i \left| T_\mathrm{num}-T_\mathrm{Eq. 6}\right|,
\label{eq:errTEq6}
\end{equation}
where the sum is taken over all positions $i$ in the fin, being spaced by $\Delta x$ and the bar representing the mean temperature of the fin. 

For an application perspective, the error in the heat flux between the plates is more relevant, which can be  expressed as $q_\mathrm{num}/q_\mathrm{Eq.10}$. We now vary the conductivity ratio $k\equiv \hat k_g/\hat k_s$ as well as the slenderness of the fins $W/L$, while keeping the ratio between $W$ and $D$ constant. Figure \ref{fig:performance} shows the performance of the model varying $k=[10^{-5},10^{-1}]$ and $\hat{\mathcal{L}}=[ 3.2, 32]$ mm. The numerical data was evaluated in the centre of the hot and the cold fin after which the difference was fitted by \autoref{eq:Thetasolve} with $C$ and $\Theta_0$ as fitting parameters. For large $C$, the data was fitted in log-space for numerical stability. The results are presented in \autoref{fig:performance}, where the error in $C$ is the ratio between $C$ (as defined in \autoref{sec:math}) and the fit. We obtain good agreement over the complete phase space.  The error in $\Theta_0$ was obtained by dividing the numerical fitting parameter by the value from the analytical model developed in \autoref{sec:math}. We find good agreement for the parameter space where $\Theta_0>-0.05$.  The analytical model however starts to underestimate $\Theta_0$ for $C>20$. The deviation seems to be more strongly related to the cooling, \textit{i.e.}, high $k$ then the slenderness $L/(W-D)$ of the fins. The two panels on the right show the error in the heat flux, deviating only more than 5\%  for $\Theta_0>-0.05$ as well.  Evaluating \autoref{eq:errTEq6} shows  that the mean difference between the numerical temperature profiles and the analytical model is smaller than 3\% for the complete phase space studied. The good agreement found between the analytical model and the numerical results demonstrate the robustness as well as the accuracy of our one dimensional model.

\begin{figure*}
\centering
\includegraphics[width=\textwidth]{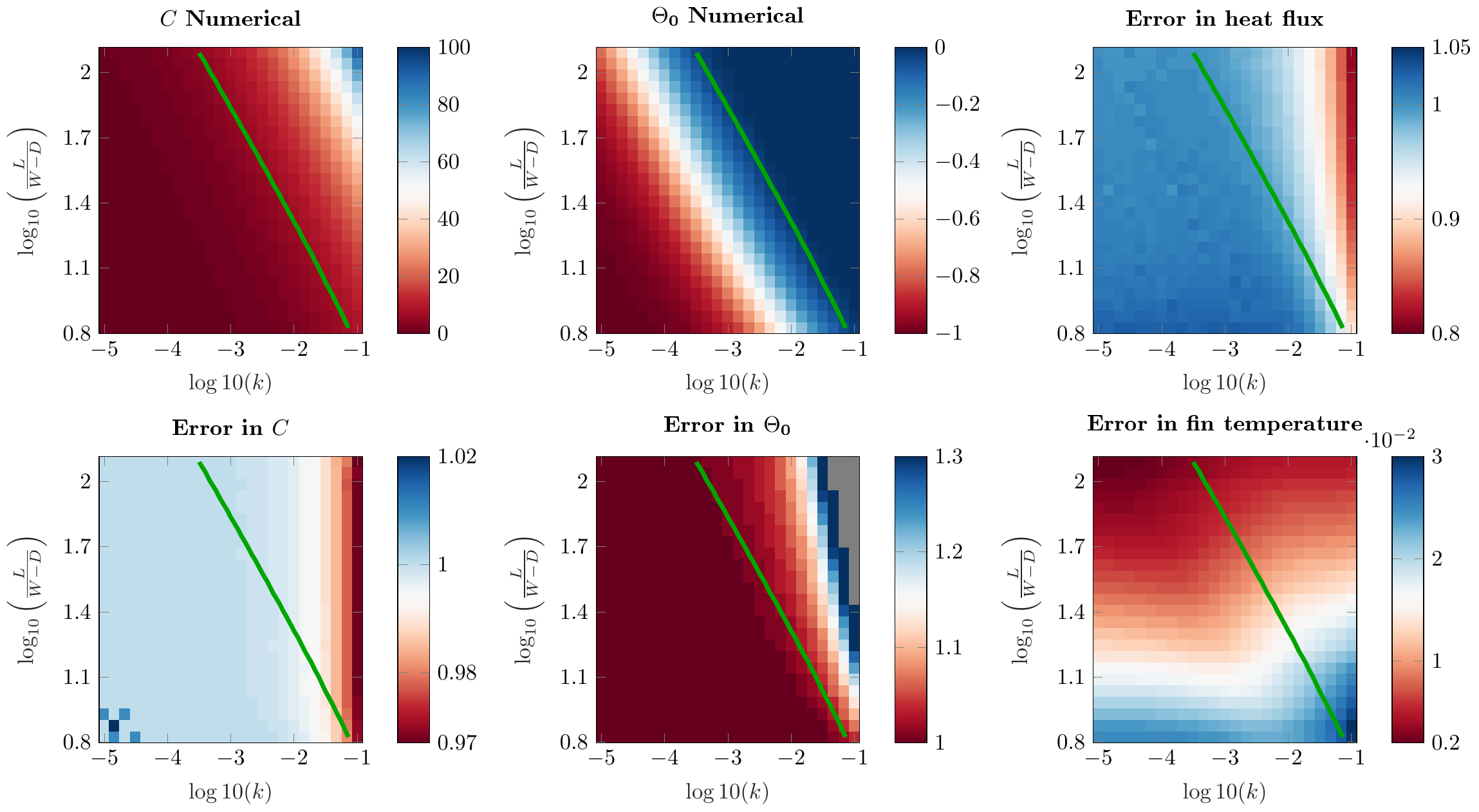}
\caption{The accuracy of the analytical model was tested for $\hat{ \mathcal{L}}=[ 3.2, 32]$ mm, $\hat \delta = 0.7$ mm, $\hat W= 0.45 $ mm and  $\hat D = 0.2 $ mm, from which $C$ (top left panel) and $\Theta_0$ (top middle panel) are presented. Errors are calculated using the methods in the text. Values of $-\Theta_0<0.05$ are left of the green lines, whereas the gray area denotes errors  exceeding 1.3.}
\label{fig:performance}
\end{figure*}

\section{Discussion}
In the previous section we developed and validated an analytical model to calculate the temperature profile in the fins of a heat switch and gave predictions for the heat flux. The parametric study for $D, \hat{\mathcal{L}}, W, \hat{k}_g$ and $\hat{k}_s $ shows that non-isothermal effects can be understood by the use of $C$. The fit \autoref{eq:theta0fit}  was able to capture the behaviour of $\Theta_0$ up to a single fitting constant: $b$. We found that the mesh generated for numerical study became too coarse for $\hat{\mathcal{L}}>8$, whose results were therefore neglected in the fitting procedure. The fit is of great value in a quick evaluation of the heat exchange across the heat switch: Equation \ref{eq:dimflux} can now be evaluated by the system parameters, $\Delta T$ and \autoref{eq:theta0fit}.

Let us briefly comment here on the fit parameter $b$. Additional simulations of the problem made clear that $b$ is a function of the thickness of the material $\delta$. By fixing the fin geometry and varying $C$ through $k$, we found that $b(0)=-0.4$ for $\delta \rightarrow 0$ and more general 
 \begin{equation}
 b\approx\frac{-0.42 }{1-1.7\delta},
 \label{eq:b}
 \end{equation}
valid for $\delta=[0, 15]$. This results in a shift of the function along the $C$- axis since $C$ is based on the complete domain. Basing $C$ on the fin length does however not eliminate the $\delta$ dependency of $b$ completely as the thickness $\delta$ of solid material itself acts as an insulating layer. For practical applications therefore, one wants to minimize the thickness $\delta$ for this reason and use $b\approx -0.4(1+1.7 \delta)$ as an approximation of \autoref{eq:b} for moderate $\delta$.

From the validation study presented in \autoref{fig:performance} of the analytical model we obtain good agreement in the regions of significant exchange between the fins. First of all, the fitted value for $C$ and $\Theta_0$ agrees within a few percent with the analytical expression (see \autoref{eq:Theta}), as well as the deviations in the temperature profiles. We encounter only strong deviations for $C>50$ or $k>0.05$. In those situations however, the effectiveness of the fins is greatly reduced anyway: The length of the fins act in fact as an insulating layer, negating the effect of the additional surface area is in this way. Our model performs thus best in the parameter space, that is relevant to the problem. We speculate that deviations occur by two reasons: First, for constant $C$, less slender fins suffer from cooling in the perpendicular directions and the one-dimensional model becomes less accurate. Secondly, for increasing $k$, the isotherms in the gap become less parallel to the side walls of the fin, see \autoref{fig:plotexample}b. As a result, the flux lines $-k\nabla T$ are no longer normal to the wall,  which is not included in the analytical model. Despite these two effects the performance of the analytical model is still good.

Finally we compare our model with the experimental measurements performed in our lab at ambient \cite{krielaart2015compact} and cryogenic \cite{vanapalli2016cryogenic} temperatures. Let us briefly elaborate on the relevant details of the setup. The heat switch is constructed out of a sintered titanium alloy (Ti$_6$Al$_4$V grade 5), having 49 rows of interleaving fins. The switch operates between the 'On'-state, where a high heat transfer is required, and the 'Off'-state by removing the gas in the gas gap. Although the gas conductivity is at most weakly dependent on the ambient pressure for bulk values, it can drastically be reduced as a result of the confinement by the fins. For pressures below 100 Pa the gas must be treated as a Knudsen gas. Krielaart et al.  \cite{krielaart2015compact} studied helium, nitrogen and hydrogen gas at room temperature ($T_m=\SI{294}{\kelvin}$), whereas Vanapalli et al. \cite{vanapalli2016cryogenic} studied helium and nitrogen at cryogenic conditions. For the cryogenic measurements the setup was operating at an mean temperature of  $\SI{117.5}{\kelvin}$, resulting in a bulk lower conductivity of the gas, as well as that of the solid compared to the ambient case. 

As shown in both references,  proper modelling of the full experimental setup required the inclusion of a contact resistance at both sides of the heat switch, see for instance \cite{yovanovich1973effect}.  Graphite foil of \SI{350}{\micro \meter} and thermal paste were added to provide good thermal contact with the segments controlling the hot and cold boundary conditions \cite{vanapalli2016cryogenic}. To obtain good agreement with the model, the total contact resistance was found to be around \SI{1/2e3}{\kelvin \per \watt}, which can be estimated using $\ell/k_{c}$. Here $k_c=0.2$ is the conductivity of the thermal paste (Apizon N) and $\ell\approx \SI{100}{\micro \meter}$ the typical layer thickness. This estimate is reasonable as it represents the roughness of the heat switch, which was made out of sintered  titanium alloy grains having a diameter of \SI{100}{\micro \meter}. Since graphite is a good thermal conductor it does not contribute to the contact resistance. 

The comparison of the experimental measurements with the analytical model is presented in \autoref{fig:Exp}. Both the helium and hydrogen measurements show good agreement with the model, whereas the largest discrepancy was found for nitrogen measurements, being approximately 10\%. The thermal contact resistance resulted in the reduction of the total heat transfer of 40\% for the highest gas conductivities measured: indeed, for those cases the thermal resistance of the heat switch becomes of the same order as the contact resistance, resulting in a significant temperature drop across the thermal paste layer.
\begin{figure}
\centering
\includegraphics[width=\columnwidth]{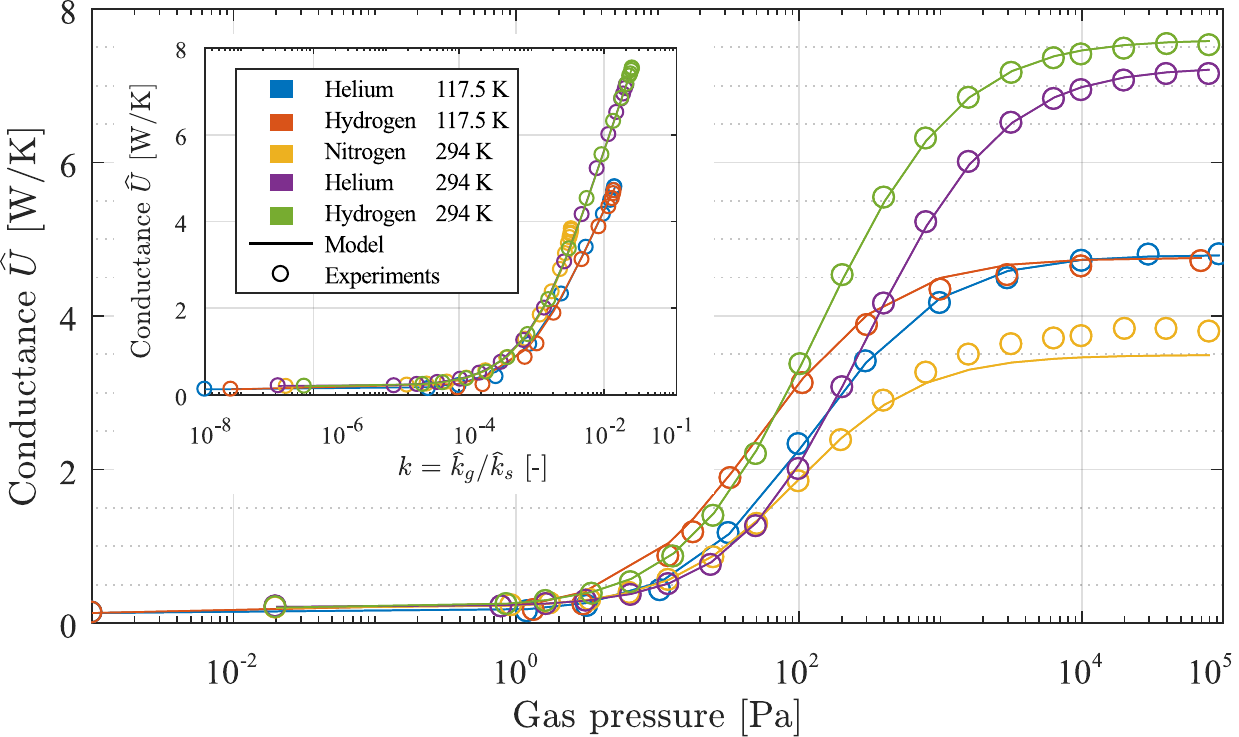}
\caption{Comparison of the model (solid lines) with experiments. The conductivity ratio $k$ is varied by reducing the pressure in the heat switch, for which the system is subject to Knudsen effects.}
\label{fig:Exp}
\end{figure}

\section{Conclusion and Outlook}

%
%
%
%
%

We found solutions to the interleaving fins where the fins may develop thermal gradients. 
The performance of the fin strongly depend whether or not this effect takes place. 
We found that the non-dimensional parameter  $C=\left(\frac{4 \hat{k}_g\hat{\mathcal{L}}^2 }
{\hat{k}_s\hat{D}(\hat{W}-\hat{D})}\right)^{1/2}$ 
characterizes this well, collapsing numerical solutions of the problem when varying $C$ 
through the different parameters over six orders of magnitude. For $C<1$ thermal gradients may be neglected and one should use \autoref{eq:isothermalflux}. In the other cases, we developed and tested an analytical model which showed excellent agreement with both numerical solutions to the problem for a large parameter space as well as with experiments. Our work offers analytical solution to find the heat flux of this type of  heat exchanger and gives new insights, essential in designing and optimizing such systems once the thermal properties are known. \\

While we studied the steady state behaviour, let us finally briefly discuss the expected behaviour of the transient operation. 
When the system of interleaving fins is used in a heat switch for switching between `on' and  `off' states, the corresponding cooling and heating of the fins themselves become relevant.
Intuitively, one can expect the off state  to be in the isothermal regime. Recovery from an previous `on' state is thus determined by heat diffusing from the hot or cold side to the tip of a fin. The (dimensional) diffusive timescale is then $\hat{\tau}_{rec}=\hat{L}^2 \hat{\rho} \hat{c}_p /\hat{k}_s$, where $ \hat{\rho}$ and $\hat{c}_p$ are the  density and specific heat of the solid. When the switch is turned `on' again, heat is first exchanged locally, for which we use the thermal timescale for transient heat transfer: $\hat{\tau}_{trans}=\hat{k}_s  \hat{\rho} \hat{c}_p/ (\hat{k}_g/\hat{D})^2$ \cite{Bejan}. It is interesting to look at the ratio between the two for which we find  $(\hat{k}_g/\hat{k}_s \cdot \hat{D}/\hat{L})^2$.  We find then that this gives a ratio of time-spans for cycling between on and off states, provided that the fluid-flow time scale is fast. This is not always ensured as viscous effect can start to play a role for small length scales. This trade off is important if one is interested in a transient operation of the fins.

\bibliography{fin}
\bibliographystyle{ieeetr}

\end{document}